\documentclass[11pt,a4paper,reqno]{amsproc}

\makeatletter
\usepackage{datetime}
\usepackage[pdfborder={0 0 0}]{hyperref}
\usepackage{upref,cases}
\hypersetup{citebordercolor=.1 .6 1}

\usepackage[T1]{fontenc}
\usepackage[utf8]{inputenc}

\usepackage[final]{graphicx}

\usepackage{geometry}
\IfFileExists{mymtpro2.sty}{\usepackage[subscriptcorrection]{mymtpro2}

  \def\bigcup{\bigcupprod}
  
  \def\bigcupdisjoint{\mathop{\kern10pt\raisebox{4pt}{$\cdot$}\kern-12pt\bigcup}\limits}
}%
	{\usepackage{amssymb,bm,dsfont}

}

\numberwithin{equation}{section}
\newtheoremstyle{ttheorem}%
       {1.8ex\@plus1ex}                
       {2.1ex\@plus1ex\@minus.5ex}      
       {\itshape}           
       {0pt}                   
       {\bfseries}          
       {.}                  
       {.5em}               
       {}                

\newtheoremstyle{ddefinition}%
       {1.8ex\@plus1ex}                
       {2.1ex\@plus1ex\@minus.5ex}      
       {}           
       {0pt}                   
       {\bfseries}           
       {.}                  
       {.5em}               
       {}                

\newtheoremstyle{rremark}%
       {1.8ex\@plus1ex}                
       {2.1ex\@plus1ex\@minus.5ex}      
       {\normalfont}        
       {0pt}                   
       {\bfseries}           
       {.}                  
       {.5em}               
       {}                   

\theoremstyle{ttheorem}
\newtheorem{theorem}{Theorem}[section]
\newtheorem{lemma}[theorem]{Lemma}
\newtheorem{proposition}[theorem]{Proposition}
\newtheorem{corollary}[theorem]{Corollary}

\theoremstyle{ddefinition}

\theoremstyle{rremark}
\newtheorem{remark}[theorem]{Remark}
\newtheorem{myremarks}[theorem]{Remarks}
\newtheorem{myexamples}[theorem]{Examples}

    {\end{nummer}\end{myremarks}}
    {\end{nummer}\end{myexamples}}

\newcounter{numcount}
\newcommand{\labelnummer}{(\roman{numcount})}%

\makeatletter
\providecommand{\showkeyslabelformat}[1]{\relax}        
\let\mysaveformat\showkeyslabelformat                   %
\def\myformat#1{\raisebox{-1.5ex}{\mysaveformat{#1}}}   %

\newenvironment{nummer}%
  {\let\curlabelspeicher\@currentlabel%
    \begin{list}{\textup{\labelnummer}}%
      {\usecounter{numcount}\leftmargin0pt%
        \topsep0.5ex\partopsep2ex\parsep0pt\itemsep0ex\@plus1\p@%
        \labelwidth2.5em\itemindent3.5em\labelsep1em%
      }%
    \let\saveitem\item%
    \def\item{\saveitem%
      \def\@currentlabel{\curlabelspeicher\kern.1em\labelnummer}}%
    \let\savelabel\label%
    \def\label##1{{\ifnum\thenumcount=1\let\showkeyslabelformat\myformat\fi\savelabel{##1}}%
										{\def\@currentlabel{\labelnummer}%
									 	\let\showkeyslabelformat\@gobble
									 	\savelabel{##1item}%
										}%
	   							}%
  }{\end{list}}%

  {\let\curlabelspeicher\@currentlabel%
    \begin{list}{\textup{\labelnummer}}%
      {\usecounter{numcount}\leftmargin0pt%
        \topsep0.5ex\partopsep2ex\parsep0pt\itemsep0ex\@plus1\p@%
        \labelwidth2.5em\itemindent0em\labelsep1em%
        \leftmargin2.5em}%
    \let\saveitem\item%
    \def\item{\saveitem%
      \def\@currentlabel{\curlabelspeicher\kern.1em\labelnummer}}%
    \let\savelabel\label%
    \def\label##1{{\ifnum\thenumcount=1\let\showkeyslabelformat\myformat\fi\savelabel{##1}}%
										{\def\@currentlabel{\labelnummer}%
									 	\let\showkeyslabelformat\@gobble
									 	\savelabel{##1item}%
										}%
    							}%
  }{\end{list}}%

\def\section{\@startsection{section}{1}%
  \z@{1.3\linespacing\@plus\linespacing}{.5\linespacing}%
  {\normalfont\bfseries\centering}}

\def\subsection{\@startsection{subsection}{2}%
  \z@{.8\linespacing\@plus.5\linespacing}{-1em}%
  {\normalfont\bfseries}}

\def\nlsubsection{\@startsection{subsection}{2}%
  \z@{.8\linespacing\@plus.5\linespacing}{.1ex}%
  {\normalfont\bfseries}}

\let\@afterindenttrue\@afterindentfalse%

\renewenvironment{proof}[1][\proofname]{\par \normalfont
  \topsep6\p@\@plus6\p@ \trivlist 
  \item[\hskip\labelsep\scshape
    #1\@addpunct{.}]\ignorespaces
}{%
  \qed\endtrivlist
}

\def\ps@firstpage{\ps@plain
  \def\@oddfoot{\normalfont\scriptsize \hfil\thepage\hfil
     \global\topskip\normaltopskip}%
  \let\@evenfoot\@oddfoot
  \def\@oddhead{
    \begin{minipage}{\textwidth}
      \normalfont\scriptsize
      \emph{\insertfirsthead}
    \end{minipage}}
  \let\@evenhead\@oddhead 
}

\def\insertfirsthead{}

\def\@cite#1#2{{%
 \m@th\upshape\mdseries[{#1}{\if@tempswa, #2\fi}]}}
\newcommand{\N}{\mathbb{N}}

\newcommand{\R}{\mathbb{R}}

\renewcommand{\le}{\leqslant}
\renewcommand{\ge}{\geqslant}

\providecommand{\bigcupdisjoint}{\mathop{\kern7pt\raisebox{6pt}{$\cdot$}\kern-9.5pt\bigcup}\limits}

\let\<\langle
\let\>\rangle

\providecommand{\norm}[1]{\lVert#1\rVert}

\providecommand{\bigparens}[1]{\bigl(#1\bigr)}
\providecommand{\Bigparens}[1]{\Bigl(#1\Bigr)}

\newcommand{\Oh}{\mathrm{O}}
\newcommand{\oh}{\mathrm{o}}

\newcommand{\upd}{\mathrm{d}}
\renewcommand{\d}{\upd}   
\newcommand{\dx}{\d x}

\let\textdef\textit
\newcommand{\hairspace}{\kern .04167em}

\def\clap#1{\hbox to 0pt{\hss#1\hss}}

\def\bra{\makeatletter\@ifstar\@bra\@@bra}
\def\@bra#1{\hairspace #1\>}
\def\@@bra#1{\lvert\@bra{#1}}
\def\ket{\makeatletter\@ifstar\@ket\@@ket}
\def\@ket#1{\<#1\hairspace}
\def\@@ket#1{\@ket{#1}\rvert}

\makeatother
\begin{document}

\title{Finite-size energy of non-interacting Fermi gases}

\author[M.\ Gebert]{Martin Gebert}

\address{Mathematisches Institut,
  Ludwig-Maximilians-Universit\"at M\"unchen,
  Theresienstra\ss{e} 39,
  80333 M\"unchen, Germany}

\email{gebert@math.lmu.de}

\thanks{Work supported by SFB/TR 12 of the German Research Council (DFG)}

\begin{abstract}
We study the asymptotics of the difference of the ground-state energies of 
two non-interacting $N$-particle Fermi gases in a finite volume of length $L$ 
in the thermodynamic limit up to order $1/L$. We are particularly interested in 
subdominant terms proportional to $1/L$, called finite-size energy. 
In the nineties Affleck and co-authors 
\cite{affleck1997boundary,zagoskin1997fermi}
claimed that 
the finite-size energy is
related to
the decay exponent 
occurring in Anderson's orthogonality. 
We prove that the finite-size energy depends on the details of the thermodynamic limit
and is therefore non-universal. Typically, it includes
an additional linear term in the scattering phase shift.
\end{abstract}

\maketitle

\section{Introduction}
Given two non-interacting $N$-particle Fermi gases, which differ by a local scattering potential, and are confined to the finite interval $(0,L)\subset (0,\infty)$, 
one can ask for two intimately connected asymptotics. 
The first one is the asymptotics of the scalar product of the two ground states $\boldsymbol\langle \Phi^N_{L},\Psi^N_{L}\boldsymbol\rangle$,
which we call the ground-state overlap in the sequel.
The second related question is
 the asymptotics of the difference of the ground-state energies $E'^N_{L}-E_{L}^N$. Both in the thermodynamic limit at some given Fermi energy $E$, 
i.e. $N/L\to\rho(E)>0$. Here, $ \rho$ is the integrated density of states of the unperturbed one-particle Schr\"odinger operator. 
These asymptotics are related to physical situations where a sudden change by a static scattering potential occurs, e.g. the Fermi edge singularity or the Kondo effect,
see \cite{affleck1994fermi}.

On the one hand, \cite{PhysRev.164.352} claims in the case of a Dirac-$\delta$ perturbation
that the ground-state overlap vanishes as
\begin{equation}\label{Anderson}
\boldsymbol\langle \Phi^N_{L},\Psi^N_{L}\boldsymbol\rangle \sim L^{-\zeta(E)/2},
\end{equation}
where 
\begin{equation}\label{123}
  \zeta(E)
    :=
    \frac{1}{\pi^2}\delta^2(\sqrt E)
\end{equation}
and $\delta$ is equal to the s-wave scattering phase shift. 
For a proof of this see \cite{Geb}. It turns out that 
the decay exponent $\zeta(E)$ is independent of the particular thermodynamic limit chosen, at least in  
the case of a Dirac-$\delta$ perturbation. 
In more general settings only upper bounds on the ground-state overlap are known, see \cite{KuOtSp13, GKM, GKM2, FrankPush, magnetic}.
In the physics literature the behaviour \eqref{Anderson} is referred to as Anderson's orthogonality catastrophe. 

On the other hand, restricting ourselves to systems on the half-axis and the family of thermodynamic limits
\begin{equation}\label{11}
 \frac N L+ \frac a L=\rho(E),
\end{equation}
where $a\in\R$ is a parameter, 
the difference of the ground-state energies admits the asymptotics
\begin{equation}\label{equation8}
 E'^N_L-E_{L}^N = \int_{-\infty}^{E}\d x\, \xi(x) +\frac {\sqrt E\pi} L x^a_{FS}(E) + \oh\Bigparens{\frac 1 L}
\end{equation}
as $N,L\to\infty$ such that \eqref{11} holds.
Here, $\xi$ is the spectral shift function for the pair of the corresponding infinite-volume one-particle Schr\"odinger operators.
In the physics literature the first term is sometimes called the Fumi term and $x^a_{FS}(E)$ 
the finite-size correction or energy, see \cite{affleck1997boundary}.
 For models on the half line with a local perturbation,
the finite-size correction $x^a_{FS}(E)$ appearing in the energy difference is claimed to be closely related to the decay exponent $\zeta(E)$ 
 occurring in Anderson's orthogonality, 
see  \cite{affleck1997boundary,affleck1994fermi,zagoskin1997fermi}. 

 In this note we give a short and elementary proof of the correct asymptotics 
 of the difference of the ground-state energies for systems on the half axis which differ by a short-range scattering potential in 
 the thermodynamic limit, see Theorem \ref{Thm:finite-size}. The proof also applies for a perturbation by a Dirac-$\delta$ perturbation. 
 It turns out that the finite-size energy $x^a_{\text{FS}}(E)$ 
 is non-universal and depends 
 on the particular choice of the parameter $a$ in the thermodynamic limit in \eqref{11}.
  Moreover, there is precisely one choice of the particle number and system size, i.e. $a= 1/2$ in \eqref{11}, such that the finite-size energy 
is equal to the Anderson exponent \eqref{123}. This particular choice was also used in a computation of the finite-size energy in \cite[App. A]{zagoskin1997fermi}.
However, for other choices of the thermodynamic limit
 an additional linear term in the  spectral shift function, or equivalently in the scattering phase shift occurs,
see Corollary \ref{Cor:finite-size} below. In contrast, it is proved in \cite{Geb} that the decay exponent $\zeta(E)$ in \eqref{Anderson} is independent of the choice of the constant $a$ in the thermodynamic limit \eqref{11}. Thus, we doubt a fundamental connection between the finite-size energy \eqref{equation8} and the 
decay exponent in Anderson's orthogonality \eqref{Anderson}.

\section{Model and results}

We consider a measurable non-negative potential $V\geq 0$ on the half line $(0,\infty)$ satisfying
\begin{equation}\label{assmp:V}
 \int_0^\infty \d x\, V(x)\left(1+x^2\right)<\infty.
\end{equation}
Moreover, let $L>1$ and $-\Delta_L$ be the negative Laplacian on the interval $(0,L)$ with Dirichlet boundary conditions. 
Then, we define the finite-volume one-particle Schr\"odinger
operators 
\begin{equation}
H_L := -\Delta_L \qquad\text{and}\qquad H_L' := -\Delta_L + V.
\end{equation}
Here, $V$ is understood as its canonical restriction to the interval $(0,L)$. 
These operators are densely defined and self-adjoint operators on the Hilbert 
space $L^2((0,L))$. Both have compact resolvents and thus admit an ONB of eigenfunctions. We denote
the corresponding non-decreasing sequences of eigenvalues, counting multiplicities, by
 $\lambda_1^L\le\lambda_2^L\le\dotsb$ and $\mu_1^L\le\mu_2^L\le\dotsb$. 
 Note that $\lambda_n^L=\left(\frac{n\pi}{L}\right)^2$, $n\in\N$, see e.g. \cite{MR0493421}. 
 Moreover, we write $H:=-\Delta$ and $H':=-\Delta+V$ for the corresponding infinite-volume operators on $L^2((0,\infty))$
 with Dirichlet boundary conditions at the origin.

 Given $N\in\N$, the induced (non-interacting)
finite-volume fermionic $N$-particle Schr\"odinger operators $\hat{H}_L$ and
$\hat{H}_L'$ act on the totally antisymmetric subspace
$\bigwedge_{j=1}^N L^2((0,L))$ of the $N$-fold tensor product space
and are given by
\begin{equation}
  \hat{H}_L^{(\prime)}
  :=
  \sum_{j=1}^N
  I\otimes\dotsm\otimes I \otimes H_L^{(\prime)} \otimes I\otimes\dotsm \otimes I,
\end{equation}
where the index $j$ determines the position of $H_L^{(\prime)}$ in the
$N$-fold tensor product of operators. 
The corresponding ground-state energies 
are given by the sum of the $N$ smallest eigenvalues
\begin{equation}
  E^N_L
  :=
  \sum_{k=1}^N \lambda^L_k
  \quad\text{and}\quad
  E'^{N}_L :=
  \sum_{j=1}^N \mu^L_j.
\end{equation}
We are interested in the difference of the ground state energies in the thermodynamic limit 
 at a given \textdef{Fermi energy} $E>0$. Thus, given $E>0$ and the number of particles $N\in\N$, 
we choose the system length $L$ such that
\begin{equation}\label{def.thermo}
\frac{N}{L}\to\frac{\sqrt E}{\pi}=\rho(E),
\end{equation}
where
 $\rho$ is the integrated density of states of the infinite-volume operator $H$.

For $k>0$ we denote by $\delta(k)$ the scattering phase shift corresponding to the pair of operators $H$ and $H'$
at the energy $k^2>0$.
Since $V\geq 0$, the phase shift is non-positive, i.e. for $k>0$
\begin{equation}
\delta(k) \leq 0.
\end{equation}
Then, the scattering matrix for the pair $H$ and $H'$ at the energy $E$ equals 
$S(E)=\exp{\big(2i\delta(\sqrt{E})\big)}$. Note that on the half line, the scattering matrix is just a complex number of modulus $1$.
For a definition of the phase shift see e.g. Appendix \ref{finite-size6}, \cite[Chapter. XI.8]{MR529429} or \cite{calogero1967variable}. 

\begin{remark}\label{Remark}
 \begin{enumerate}
  \item[(i)] Let $\xi$ be the spectral shift function for the pair of operators $H$ and $H'$. Then, we have 
        the identity \cite{zbMATH00232645}
	\begin{equation}
	\displaystyle\frac 1 \pi \delta(\sqrt E)= -\xi(E),
	\end{equation}
	for every $E>0$. 
  \item[(ii)]  We define for $E>0$
  \begin{equation}
   \zeta(E) :=\frac 1 {\pi^2}\delta^2(\sqrt E).
  \end{equation} 
    This constant
	       equals the decay exponent found in \cite{Geb} which determines the asymptotics of the exponent in Anderson's orthogonality, i.e.
	       the asymptotics \eqref{Anderson} of the scalar product of the ground states of 
	       the pair of operators $\hat{H}_L$ and $\hat{H}_L^{(\prime)}$ in the thermodynamic limit.
 \end{enumerate}
\end{remark}

Using the notation of Remark \ref{Remark}, our main 
 result is the following:

\begin{theorem}\label{Thm:finite-size}
For all Fermi energies $E>0$ the difference of the ground-state energies admits the asymptotics
\begin{align}
 E'^{N}_{L} -E^{N}_{L}&=-\frac 1 \pi \int_{0}^{\left(\frac {N\pi} L\right)^2} \d x\, \delta(\sqrt x)+
   \frac { \sqrt E} {L} \Big(-\delta(\sqrt E)+\frac 1 \pi \delta^2(\sqrt E) \Big)+ \oh\Bigparens{\frac 1 {L} }\notag\\
   &=\int_{0}^{E} \d x\,\xi(x) + \int_{E}^{\left(\frac {N\pi} L\right)^2} \d x\,\xi(x) +
    \frac { \sqrt E\pi} {L} \big(\xi(E) + \zeta(E) \big) +\oh\Bigparens{\frac 1 {L} }\label{111}
\end{align}
as $N, L\to\infty$, and $\frac N L \to\frac{\sqrt E} \pi$.  
\end{theorem}

\begin{remark}\label{remark3}
Since $\xi$ is continuous, see Lemma \ref{Le:finite-size3} below, 
\begin{equation}
 \int_{E}^{\left(\frac {N\pi} L\right)^2} \d x\,\xi(x)= \bigg(\Big(\frac {N\pi} L\Big)^2 -E\bigg)\xi(E)+ 
 \oh\Big(\Big(\frac {N\pi} L\Big)^2 -E\Big)
\end{equation}
as $N,L\to\infty$, and $\frac N L \to \frac{\sqrt E}\pi>0$.
This immediately implies that the asymptotics depends on the rate of convergence of the thermodynamic limit
and that the finite-size energy defined in \eqref{equation8} is non-universal. In general, the first-order correction to the difference of the ground-state energies may even be $L$ dependent. 
\end{remark}
Remark \ref{remark3} implies for the particular family of
 thermodynamic limits considered in the introduction:

 \begin{corollary}[Finite-size energy]\label{Cor:finite-size}
    For a given Fermi energy $E>0$, some particle number $N\in\N$ and $a\in\R$ 
    we choose the system length $L$ such that
\begin{equation}\label{def:thermo}
 \frac{N+a}{L}:=\frac{\sqrt E} \pi.
\end{equation}
    Then, the $1/L$-correction in \eqref{111}, which is called the finite-size energy introduced in \eqref{equation8}, is  
    \begin{equation}
     x^a_{FS}(E)=   \big(1-2a\big)\xi(E)+ \zeta(E). 
    \end{equation}
Thus, 
\begin{enumerate}
 \item [(i)] for the particular choice $a=\frac 1 2$ the finite-size energy is
	      \begin{equation}
	       x_{FS}(E)=  \zeta(E),
	      \end{equation}
 \item[(ii)] whereas for the choice $a=0$ the finite-size energy equals
	      \begin{equation}
	       x_{FS}(E)= \xi(E)+\zeta(E).
	      \end{equation}
\end{enumerate}
 \end{corollary}

\begin{remark}\label{remark2}
\begin{enumerate}
  \item [(i)]  In our case of $V\geq 0$ the integrals in Theorem \ref{Thm:finite-size} may start from $0$, since $\delta(x)=0$ for $x\leq 0$. 
  \item [(ii)]The first term in the expansion is not surprising since
	      \begin{align}
	       E^{'N}_{L} -E^{N}_{L}=\int_{-\infty}^E \dx\, \xi_L(x) + \oh\left(1\right),
	      \end{align}
	      where $\xi_L$ is the finite-dimensional spectral shift function and 
	      \begin{equation}
	       \int_{-\infty}^E\d x\,\xi_L(x)\to \int_{-\infty}^E\d x\, \xi(x)\
	      \end{equation}	 
	      as $L\to\infty$,
	      see \cite{MR2596053} or \cite{MR2892556} 
	      for definitions and details. 
  \item [(iii)] The same result with the completely analogous proof holds also for a Dirac-$\delta$ perturbation or
                s-wave scattering in three dimensions which is considered in \cite{Geb}. In the special case of the Neumann and Dirichlet Laplacian $H:=-\Delta^N$ 
		and $H':=-\Delta^D$ on $L^2((0,\infty))$ 
		the proof is even simpler since the phase shift is energy independent 
		\begin{equation}
		\delta(\sqrt E)=\frac \pi 2
		\end{equation}
		and one easily obtains the $a$-dependence in Corollary \ref{Cor:finite-size}.  
  \item [(iv)]  We choose $V\geq0 $ since we want to avoid bound states or zero-energy resonances. 
		Moreover, the integrability assumption \eqref{assmp:V} on $V$ ensures sufficient regularity of the phase shift $\delta$. 
  \item [(v)]  Our result allows also a conclusion for the same problem on $\R$ with a symmetric perturbation $V$ because in this case
		the problem is reduced to two problems on the half axis
                with either Neumann or Dirichlet boundary condition at the origin.
      
\end{enumerate}
\end{remark}

\section{Proof of Theorem \ref{Thm:finite-size}}

We start with a lemma relating the eigenvalues of 
the pair of finite-volume operators.
 
\begin{lemma}\label{Le:finite-size1}
Let $\delta$ be the phase shift for the pair of operators $H$ and $H'$ then the $n$th eigenvalues of $H_L$ and $H'_L$ satisfy 
   \begin{equation}
   \sqrt{\mu_n}=\sqrt{\lambda_n}- \frac{\delta(\sqrt{\mu_n})} L+\oh\Bigparens{\frac 1 {L^2}}, 
   \end{equation}
   where the error depends only on the potential $V$.
\end{lemma}

This follows directly from introducing Pr\"ufer variables in the theory of Sturm-Liouville operators. 

We have to investigate the behaviour of $\delta$ at $k=0$ to obtain suitable error estimates.

\begin{lemma}\label{Le:finite-size3}
Let $\delta$ be the phase shift corresponding to the operators $H$ and $H'$. Then,
$\delta\in C^2((0,\infty))$ and there exists a constant $c$, depending on the potential $V$, such that
 for all $k>0$
\begin{enumerate}
 \item [(i)] $|\delta(k)|\leq c \min\{ k,\frac 1 { k}\}$, in particular $\delta\in L^\infty((0,\infty))$. 
 \item [(ii)] $\delta'\in L^\infty((0,\infty))$,
 \item [(iii)] $|\delta''(k)|\leq  \frac c {k}$.
\end{enumerate}
Moreover, 
\begin{enumerate}
 \item[(iv)] we have the following expansion of the phase shift
   \begin{equation}
    \delta(\sqrt{\mu_n})=\delta(\sqrt{\lambda_n}) - \frac {\delta'(\sqrt{\lambda_n})\delta(\sqrt{\lambda_n})}  L + \frac {F(\sqrt{\lambda_n})} {L^2},
   \end{equation}
  where the remainder term obeys for $x>0$
   \begin{equation}\label{eq:3}
    \big|F(x)\big|\leq c\Bigparens{\frac 1 {x} +1}
   \end{equation}
   for some constant $c$ depending on the potential $V$. 
\end{enumerate}
\end{lemma}

Lemma \ref{Le:finite-size1} and \ref{Le:finite-size3} are well known to experts in the theory of Sturm-Liouville operators. 
Unfortunately, we did not find a precise reference. For convenience, we prove both results in Appendix \ref{finite-size20}.
The third ingredient to the proof of Theorem \ref{Thm:finite-size} is
the following:
  
\begin{lemma}\label{Le:finite-size2}{(Euler-MacLaurin)}
\begin{enumerate}
 \item [(i)] Let $f\in C^1((0,\infty))$ then 
\begin{equation}
 \frac 1 L \sum_{n= 1}^{ N} f\left(\frac n L\right)= \int_0^{\frac N L} \d x\, f(x) + \Oh\Bigparens{\frac {N} {L^2}}\norm{f'}_{L^\infty\left((0,\frac N L)\right)}.
\end{equation}
 \item [(ii)] Let $f\in C^2((0,\infty))$ with $f''\in L^\infty\left((0,\infty)\right)$ then 
\begin{equation}
 \frac 1 L \sum_{n= 1}^{ N} f\left(\frac n L\right)= \int_0^{\frac N L} \d x\, f(x)+ \frac 1 {2L} \int_0^{\frac N L} \d x\, f'(x) + 
 \Oh\Bigparens{\frac N {L^3}}.
\end{equation}
\end{enumerate}
\end{lemma}

The proof of this lemma is elementary, see also \cite[Chapter  XIV]{zbMATH00861508}.

\begin{proof}[proof of Theorem \ref{Thm:finite-size}]
Using Lemma \ref{Le:finite-size1}, 
 we obtain
\begin{align}
 \sum_{n=1}^N \big(\mu_n-\lambda_n\big)  = 
 \sum_{n=1}^N \Bigparens{-\frac {2\sqrt{\lambda_n}\delta(\sqrt{\mu_n})} L + \frac{\delta^2(\sqrt{\mu_n})}{L^2}}+
 \oh\Bigparens{\frac N {L^2}}\label{equation11}
 \end{align}
On the other hand 
Lemma \ref{Le:finite-size3}  (iv) provides
 \begin{align}
  \eqref{equation11}=& \sum_{n=1}^N  \Bigparens{-\frac{2\delta(\sqrt{\lambda_n})\sqrt{\lambda_n}} L +  \frac {2\delta'(\sqrt{\lambda_n})\delta(\sqrt{\lambda_n})\sqrt{\lambda_n}} {L^2}
     + \frac{\delta^2(\sqrt{\lambda_n})}{L^2}}\notag\\
                    &+ \frac 1 {L^3}  \sum_{n=1}^N G\big(\sqrt{\lambda_n}\big) + \oh\bigparens{\frac N {L^2}},
 \end{align}
where
\begin{align}
 G(x)=  \Big(& -2\delta'(x)\delta^2(x) - 2x F(x) +\frac 1 L \big((\delta'(x)\delta(x))^2 + 2\delta(x)F(x)\big)\notag\\
	  &-\frac 2 {L^2} F(x)\delta'(x)\delta(x)+ \frac 1 {L^3} F^2(x)\Big).
\end{align}
Since $\lambda_n= \left( \frac{n\pi}{L} \right)^2$, $\frac N L\to  \frac{\sqrt E}\pi$, using Lemma \ref{Le:finite-size3} (i)-(iii) and \eqref{eq:3}, we obtain for the error

\begin{equation}\label{eq:error2}
  \frac 1 {L^3} \sum_{n=1}^N G\big(\sqrt{\lambda_n}\big)=\Oh\Bigparens{\frac 1 {L^2}}.
\end{equation}
Note that by Lemma  \ref{Le:finite-size3} the function $f:x\mapsto x\delta(x)$ fulfills the assumptions of Lemma \ref{Le:finite-size2} (ii).
Thus, we compute

\begin{align}
  \sum_{n=1}^N & -\frac{2\delta(\sqrt{\lambda_n})\sqrt{\lambda_n}} L  = - \frac 1 L \sum_{n= 1}^{ N } 2\delta\left(\frac {n\pi} L\right)\frac {n\pi} L\notag\\
  & =  -\int_0^{\frac N L} \d x\, 2 \delta(x\pi)(x\pi) - \frac 1 L \int_0^{\frac N L} \d x\, \left(\delta(x\pi)(x\pi)\right)'+ \Oh\Bigparens{\frac N {L^3}}\notag\\
  & = - \frac 1 \pi \int_0^{(\frac {N\pi} L)^2} \d x\,  \delta(\sqrt x) - \frac 1 L \delta(\sqrt E)\sqrt E + \oh\Bigparens{\frac 1 {L}}\label{eq:finite-size1},
\end{align}
where we used in the last equality the convergence  $\frac N L\to \frac{\sqrt E}\pi$ and the continuity of $\delta$.
Using  Lemma  \ref{Le:finite-size3} we see that $g:x\mapsto x\delta(x)\delta'(x)$ satisfies the assumptions of Lemma \ref{Le:finite-size2} (i) with
$\norm{g'}_{L^\infty\left((0,\frac N L)\right)}\leq c(1+\frac N L)$. Therefore,
\begin{align}
   \sum_{n=1}^N &  \frac {2\delta'(\sqrt{\lambda_n})\delta(\sqrt{\lambda_n})\sqrt{\lambda_n}} {L^2} 
       =
   \frac 1 L\bigg(\frac 1 L  \sum_{n= 1 }^{ N }  2 \delta'\left(\frac {n\pi} L\right)\delta\left(\frac {n\pi} L \right) \frac {n\pi} L  \bigg)\notag\\ 
      & =
  \frac 1 {L} \int_0^{\frac N L} \d x\,  2\delta'(x\pi)\delta(x\pi) (x\pi) + \Oh\Bigparens{\frac N {L^3}}\Bigparens{1+\frac N L}\notag \\
      & =
   \frac 1 {L\pi} \Bigparens{ \delta^2(\sqrt E) \sqrt E - \int_0^{\frac N L}\d x\, \delta^2(x\pi)\pi} + \oh\Bigparens{\frac 1 {L}}\label{eq:finite-size2},
\end{align}
where we used integration by parts, the convergence $\frac N L\to \frac{\sqrt E}\pi$ and the continuity of $\delta$ in the last line.
Lemma \ref{Le:finite-size3} yields the assumptions of Lemma \ref{Le:finite-size2} (i) for $h:x\mapsto\delta^2(x)$ with $h'\in L^\infty\left((0,\infty)\right)$. 
Thus,
\begin{align}
  \sum_{n=1}^N  \frac{\delta^2(\sqrt{\lambda_n})}{L^2} 
      & =
   \frac 1 L \bigg( \frac 1 L \sum_{n= 1 }^{N} \delta^2\left(\frac {n\pi} L\right) \bigg)\notag\\
      & =
   \frac 1 {L} \int_0^{\frac N L}\dx\, \delta^2( x\pi) + \Oh\Bigparens{ \frac 1 {L^2}}\label{eq:finite-size3}.
\end{align}
Summing up \eqref{eq:finite-size1}, \eqref{eq:finite-size2}, \eqref{eq:finite-size3} and equations \eqref{equation11}, \eqref{eq:error2} give the claim. 
\end{proof}

\appendix
\section{Pr\"ufer variables and the phase shift}\label{finite-size6}

Let $k>0$. We consider the eigenvalue problem on $(0,\infty)$
    \begin{equation}\label{eq:finite-size}
      -u''+ Vu= k^2 u, \qquad u(0)=0.
    \end{equation}
    Introducing Pr\"ufer variables 
    \begin{equation}\label{prufer}
     u(x)= \rho_u(x) \sin(\theta_k(x)) \qquad u'(x)=k \rho_u(x)\cos(\theta_k(x)),
    \end{equation}
    \eqref{eq:finite-size} is equivalent to the system 
    \begin{align}
      \theta_k'&=k - \frac 1 {k} V\sin^2(\theta_k),\qquad \theta_k(0)=0,\label{equation6} \\
      \rho_u'&=\frac {V\sin(2\theta_k)} {2k} \rho_u\label{equation7},
    \end{align}
see e.g. \cite[Sec. 14.4]{Weidmann}. 
We call $\theta_k$ the Pr\"ufer angle.
Using the Banach fixed-point theorem, there exist absolutely continuous solutions $\theta_k$ and $\rho_u>0$ of \eqref{equation6}
and \eqref{equation7}. Moreover, the solution $\theta_k$ is unique.  
We denote by
\begin{equation}\label{pr-ph}
\delta_k(x):= \theta_k(x) - k x, \quad\text{where}\ k,x>0
\end{equation}
 the phase shift function,
Then, for $k>0$ the scattering phase shift $\delta$ is defined by
\begin{equation}\label{phaseshift2}
 \lim_{x\to\infty}\delta_k(x):=\delta(k).
\end{equation}
Therefore, 
integrating \eqref{equation6} implies
\begin{equation}\label{phaseshift2}
 \delta(k)=-\frac 1 {k} \int_0^\infty\d t\, V(t)\sin^2(\theta_k(t)). 
\end{equation}

The non-linear ODE \eqref{equation6}
is sometimes called the variable-phase equation,
 see e.g. \cite{calogero1967variable} or \cite[Thm. XI.54]{MR529429}.
We did not choose the standard Pr\"ufer variables. But with the choice \eqref{prufer} it is 
particularly easy to compare the Pr\"ufer angle with the phase-shift function and in turn with the
phase shift. This was also used in \cite{KLS}.
We continue with some elementary properties of the Pr\"ufer angle, respectively of the phase-shift function for perturbations $V\ge 0$.

\begin{proposition}\label{Pr:finite-size1}
 Let $V\ge 0$, $k>0$ and fix $x>0$. Then,
 \begin{enumerate}
  \item [(i)] $\theta_k(x)$ is non-negative, moreover,
               \begin{equation}\label{pr:eq1} 
                 0  \leq \theta_k(x) \leq  k x,
               \end{equation}
  \item [(ii)] we have
	       \begin{equation}\label{pr:eq2}
               \lim_{k\to 0} \theta_k(x)=0, \qquad \lim_{k\to\infty} \theta_k(x)=\infty,
               \end{equation}
   \item [(iii)] the functions $k\mapsto\theta_k(x)$ and $k\mapsto\delta_k(x)$ are twice differentiable, i.e.
               \begin{equation}\label{pr:eq3}
                \theta_{(\,\cdot\,)}(x),\ \delta_{(\,\cdot\,)}(x)\in C^2((0,\infty)),
               \end{equation}                          
               where $\frac{\partial}{\partial k} \theta_{k}$ 
               satisfies
               \begin{align}\label{pr:eq4}
                &\frac {\partial}{\partial k}\theta_k(x)=  \int_0^x \d t\, 
                \frac {\rho^2(t)} {\rho^2(x)}\Bigparens{1 + \frac{V(t)\sin^2(\theta_k(t))}{ k^2}}\geq 0.
               \end{align}
 \end{enumerate}
\end{proposition}

\begin{proof}[Proof of Proposition \ref{Pr:finite-size1}]
The first inequality in (i) follows from integrating equation \eqref{pr:eq4} with the initial condition $\theta_k(0)=0$ for all $k>0$.
The second inequality follows from $V\ge 0$ and the ODE \eqref{equation6}.

The first equality in (ii) is a consequence of (i). The second equality follows directly from the ODE.

Let $u(x,k)$ be a non-trivial solution of \eqref{eq:finite-size}.
Standard results provide that $u$ and $u'$ are analytic functions in the parameter $k$ \cite[Kor. 13.3]{Weidmann}. 
Note that $u$ and $u'$ do not have the same zeros. 
Since $\tan\big(\theta_k(x)= ku(x,k)/ u'(x,k)\big)$ for $u'(x,k)\neq 0$ and $\text{cotan}\,\big(\theta_k(x)\big) = u'(x,k)/(ku(x,k))$ 
for $u(x,k)\neq 0$, the properties \eqref{pr:eq3} follow from the analyticity of $u$. 
We compute
\begin{align}
 \Big(\rho^2\frac{\partial}{\partial k}\theta\Big)_x&= 
 2\rho\rho_x\frac{\partial}{\partial k}\theta +\rho^2\frac{\partial}{\partial k} \theta_{ x}\notag\\
&= 2\rho\rho_x\frac{\partial}{\partial k}\theta + \rho^2 \frac{\partial}{\partial k}\Bigparens{k- \frac{V\sin^2(\theta)}{k}}\notag\\
	&= 2\rho\rho_x\frac{\partial}{\partial k}\theta + \rho^2 \Bigparens{ 1 + \frac{V\sin^2(\theta)}{k^{2}}- \frac{V\sin(2\theta)}{ k} \frac{\partial}{\partial k} \theta}\notag\\
				&=\rho^2\Bigparens{ 1 + \frac{V\sin^2(\theta)}{k^{2}}}.							
\end{align}
Integrating the latter yields \eqref{pr:eq4}. This computation is adopted from \cite[Lem. 14.16]{Weidmann}.
\end{proof}

\begin{proof}[proof of Lemma \ref{Le:finite-size1}]\label{finite-size15}
Let $\mu>0$. Consider the eigenvalue equation on $[0,L]$
 \begin{equation}
      - u''+ Vu= \mu u, \qquad u(0)=0.
 \end{equation} 
We introduce Pr\"ufer variables according to \eqref{prufer}. 
Note that any eigenfunction $u$ of $h_L'^D$ corresponding to an eigenvalue $\mu$ has to fulfill $u(L)=0$ due to the Dirichlet boundary condition at $L$.
Thus, using $\rho_u(x)\neq 0$ for all $x\geq 0$,  we obtain 
 $\sin\big(\theta_{\sqrt\mu}(L)\big)=0$. With \eqref{pr:eq2} and \eqref{pr:eq4} this implies
for the $n$th eigenvalue $\mu_n$ of $h_L'^D$
\begin{equation}
 \theta_{\sqrt{\mu_n}}(L)=n\pi.
\end{equation}
Therefore, integrating \eqref{equation6} yields 
\begin{equation}
      \sqrt\mu_n=\frac{n\pi} L + \frac 1 {L\sqrt\mu_n } \int_0^L \d t\, V(t) \sin^2\big(\theta_{\sqrt{\mu_n}}(t)\big).
\end{equation}
Now, using $|\sin(x)|\leq |x|$, \eqref{pr:eq1}, $|\sin(x)|\leq 1$ and \eqref{assmp:V} we obtain
\begin{align}
  \frac{1}{\sqrt{\mu_n}}\int_L^\infty \d t\, V(t) \sin^2(\theta_{\sqrt{\mu_n}}(t))\leq &\int_L^\infty\d t\, V(t)t\notag\\
  \leq &\frac 1 {L} \int_L^\infty\d t\, t^2 V(t)= \oh\Bigparens{\frac 1 L}.
\end{align}
 Then, \eqref{phaseshift2} and $\sqrt{\lambda_n}=\frac{n\pi} L$ give the claim. 
\end{proof}

\begin{proof}[proof of Lemma \ref{Le:finite-size3}]\label{finite-size20}
Part (i) follows from \eqref{phaseshift2}, \eqref{pr:eq1} and \eqref{assmp:V}.

Equation \eqref{pr:eq3} implies 
\begin{equation}\label{equation10}
\frac {\partial}{\partial k}\theta_k(x)=  \int_0^x \d t\, \frac {\rho^2(t)} {\rho^2(x)}\Bigparens{1 + \frac{V(t)\sin^2(\theta_k(t))}{ k^2}}.
\end{equation}
The ODE \eqref{equation7}, \eqref{pr:eq1}, the elementary inequality $|\sin x|\leq |x|$ and  \eqref{assmp:V} imply
\begin{equation}\label{eq:1}
 \Big| \frac {\rho(t)}{\rho(x)} \Big|\leq \exp\Bigparens{\int_t^x\d s\ s V(s)}\leq \exp\big(\norm{(\,\cdot\,)V}_1\big)<\infty.
\end{equation}
From this, \eqref{pr:eq1} and $|\sin x|\leq |x|$ 
we infer the existence of a constant $c$ depending on the potential $V$ such that
\begin{equation}\label{eq:2}
\Big|\frac {\partial}{\partial k}\theta_k(x)\Big|\leq c\left(1+x\right).
\end{equation}
Then, the above, \eqref{pr:eq1} and dominated convergence provide $\delta\in C^1((0,\infty))$ with 
\begin{equation}
 \big|\delta'(k)\big| \leq  c \int_0^\infty \d t\, V(t)\big(1+t+t^2\big).
\end{equation}
The assumptions on the potential give the claim.

Using \eqref{pr:eq3}, we compute as in the proof of \eqref{pr:eq4}
\begin{equation}
\Bigparens{\rho^2\frac{\partial^2}{\partial k^2}\theta}_x =2\rho^2V\Bigparens{-\frac{\sin^2(\theta)}{k^3}+\frac{\sin(2\theta)\frac{\partial}{\partial k}\theta}{k^2}
-\frac{\cos(2\theta)(\frac{\partial}{\partial k}\theta)^2} k}.
\end{equation}
 Using \eqref{pr:eq1}, $|\sin x|\leq |x|$, \eqref{eq:1} and \eqref{eq:2}, we see
\begin{equation}\label{derivative}
 \Big|\frac {\partial^2}{\partial k^2}\theta_k(x)\Big|\leq \frac{\tilde c} {k},
\end{equation}
where $\tilde c$ depends on $V$.
Then dominated convergence yields $\delta\in C^2((0,\infty))$ and  \eqref{pr:eq1} and \eqref{derivative} provide
\begin{equation}
 \big|\delta''(k)\big| \leq \frac C {k} \int_0^\infty \d t\, V(t)\big(1+t+t^2\big)
\end{equation}
for some $C$ depending on the potential $V$. 
  
To prove (iv) we use
   Lemma \ref{Le:finite-size1}. Thus,
    \begin{equation}
     \sqrt{\mu_n}=\sqrt{\lambda_n}+\frac{\delta(\sqrt{\mu_n})} L + \oh\Bigparens{\frac 1 L},
    \end{equation}
    Since $\delta\in C^2((0,\infty))$ we compute for $x,y\in(0,\infty)$ with $y>x$ and $y=x+\frac{\delta(y)} L+ \oh\big(\frac 1 L\big)$ 
    \begin{align}
     \Bigl\lvert \delta\left(y\right)-\delta(x)+\frac{\delta'(x)\delta(x)} L\Bigr\rvert 
     &\leq  \Bigl\lvert \int_x^{y} \d t\, \int_x^t \d s \delta''(s)\Bigr\rvert+
	 |\delta'(x)|  \Bigl\lvert y- x + \frac{\delta(x)} L \Bigr\rvert\nonumber\\
	& \leq  \frac 1 {x} |y-x|^2+ \frac{\norm{\delta}_\infty } L \Big|\int_x^y\d t\, \delta'(t) + 
	\oh\Bigparens{\frac 1 L}\Big|.\label{equation5}
	 \end{align}
      Using Lemma \ref{Le:finite-size3} (ii)  and once again
     the recursion relation we obtain
      \begin{align}
	  \Bigl\lvert \delta\left(y\right)-\delta(x)+\frac{\delta'(x)\delta(x)} L\Bigr\rvert &\leq \Bigparens{\frac 1 {x}+1} \Oh\Bigparens{\frac 1 {L^2}}.
    \end{align}
    The claim follows from setting $x:=\lambda_n$ and $y:=\mu_n$.  
\end{proof}

\section*{Acknowledgement}
The author thanks Peter Otte and Wolfgang Spitzer for bringing up the problem and also Hubert Kalf and Peter M\"uller for helpful discussions.

\newcommand{\etalchar}[1]{$^{#1}$}
\newcommand{\noopsort}[1]{}


\begin{thebibliography}{HSBvD05}
\providecommand{\url}[1]{{\tt #1}}
\providecommand{\urlprefix}{URL }
\providecommand{\eprint}[2][]{e-print {#2}}



\bibitem[Aff97]{affleck1997boundary}
I. Affleck, Boundary condition changing operations in conformal field theory and condensed matter physics,
{\em Nuc. Phys. B\/} {\bf 58}, 35--41 (1997). 


\bibitem[AL94]{affleck1994fermi}
I. Affleck and A.~W. Ludwig, The Fermi edge singularity and boundary condition changing operators,
{\em J. Phys. A\/} {\bf 27}, 5375--5392 (1994). 


\bibitem[And67]{PhysRev.164.352}
P.~W. Anderson\noopsort{z}, Ground state of a magnetic impurity in a metal,
  {\em Phys. Rev.\/} {\bf 164}, 352--359 (1967).

  
\bibitem[BM12]{MR2892556}
V. Borovyk and K.~A. Makarov,
On the weak and ergodic limit of the spectral shift function,
{\em Lett. Math. Phys.\/} {\bf 100}, 1--15 (2012).  
  
  
\bibitem[BY92]{zbMATH00232645}
M.~Sh. Birman and D.~R. Yafaev, The spectral shift function. The work of M. G. Krein and its further development,
  {\em St. Petersbg. Math. J.\/} {\bf 4}, 1--44 (1992).
  
  
\bibitem[Cal67]{calogero1967variable}
F. Calogero, {\em Variable phase approach to potential scattering},
Academic Press, New York, 1967.
  

\bibitem[FP15]{FrankPush}
R.L. Frank and A. Pushnitski,
The spectral density of a product of spectral projections, {\em J. Funct. Anal.} {\bf 268}, 3867--3894 (2015).
  
  
\bibitem[GKM14]{GKM}
M. Gebert, H. K{\"u}ttler, and P. M{\"u}ller,
Anderson's Orthogonality Catastrophe,
{\em Comm. Math. Phys.\/}, {\bf 329}, 979--998 (2014).


\bibitem[GKMO14]{GKM2}
M. Gebert, H. Küttler, P. Müller, and P. Otte,
The decay exponent in the orthogonality catastrophe in Fermi gases, 
arXiv:1407.2512 (2014). To appear in {\em J. Spect. Theory}.


\bibitem[Geb15]{Geb}
M.~Gebert,
The asymptotics of an eigenfunction-correlation determinant for Dirac-$\delta$ perturbations, 
{\em J. Math. Phys.} {\bf 56}, 072110 (2015).


\bibitem[HM10]{MR2596053}
P.~D. Hislop and P. M{\"u}ller,
The spectral shift function for compactly supported perturbations of
  {S}chr{\"o}dinger operators on large bounded domains,
{\em Proc. Amer. Math. Soc.\/} {\bf 138}, 2141--2150 (2010).


\bibitem[KLS98]{KLS}
A.~Kiselev, Y.~Last and B.~Simon,
Modified Pr\"ufer variables and EFGP transforms and the spectral analysis of one-dimensional Schr\"odinger operators,
{\em Commun. Math. Phys.\/} {\bf  194}, 1--45 (1998).


\bibitem[{Kno}96]{zbMATH00861508}
K. Knopp,
 {\em Theorie und Anwendung der unendlichen Reihen}, , 6th ed. edition,
 Springer-Verlag, Berlin, 1996.

 
\bibitem[KOS13]{KuOtSp13}
H. K\"uttler, P. Otte and W. Spitzer,
Anderson's Orthogonality Catastrophe for One-Dimensional Systems,
{\em Ann. H. Poincar\'e\/}, {\bf 15}, 1655--1696 (2014). 


 \bibitem[KOS15]{magnetic}
H.~K.~Kn\"orr, P.~Otte, and W.~Spitzer,
Anderson's orthogonality catastrophe in one dimension induced by a magnetic field,
{\em J. Phys. A: Math. Theor.\/} {\bf 48}, 325202 (2015).


\bibitem[RS78]{MR0493421}
M. Reed and B. Simon,
 {\em Methods of modern mathematical physics. {IV}. {A}nalysis of
  operators},
Academic Press, New
  York, 1978.

\bibitem[RS79]{MR529429}
M. Reed and B. Simon,
{\em Methods of modern mathematical physics {III}},
 Academic Press, New York, 1979.

 
\bibitem[Wei03]{Weidmann}
J. Weidmann,
{\em Lineare Operatoren in Hilbertr\"aumen},
 Teubner, Stuttgart, 2003.
 

\bibitem[ZA97]{zagoskin1997fermi}
A.~M. Zagoskin and I. Affleck,
 Fermi edge singularities: Bound states and finite-size effects,
 {\em J. Phys. A\/} {\bf 30}, 5743--5765 (1997).
  
  

\end{thebibliography}
\end{document}